\documentclass[version=final]{iacrcc}
\license{CC-by}

\usepackage{amsmath,amssymb,amsfonts}
\usepackage{amsthm}
\usepackage{algorithmicx,xpatch}
\usepackage{algorithm}
\usepackage{algpseudocode}
\usepackage{graphicx}
\usepackage{textcomp}
\usepackage{xcolor}

\usepackage[utf8]{inputenc}
\usepackage[T1]{fontenc}
\usepackage{stmaryrd}
\usepackage{listings}
\usepackage{tikz}
\usepackage{setspace}
\usepackage{multirow}
\usepackage{url}
\usepackage{caption}
\usepackage{moresize}

\usepackage{enumitem}
\usepackage{pifont}
\usepackage{booktabs}
\usepackage{etoolbox}
\usepackage{verbatim}
\usepackage{float}
\usepackage{ifthen}
\usepackage{scalefnt}
\usepackage{blkarray}
\usepackage{subfig}
\usepackage[mathscr]{eucal}
\usepackage{accents}
\usepackage{booktabs}\let\cline\cmidrule

\newcommand{\mavxi}{\texttt{\_\_m512i}~}

\theoremstyle{definition}

\theoremstyle{remark}

\newcommand{\uwidehat}[1]{%
  \mathpalette\douwidehat{#1}%
}

\makeatletter
\newcommand{\douwidehat}[2]{%
  \sbox0{$\m@th#1\widehat{\hphantom{#2}}$}%
  \sbox2{$\m@th#1x$}
  \sbox4{$\m@th#1#2$}
  \dimen0=\ht0
  \advance\dimen0 -.8\ht2
  \dimen2=\dp4
  \rlap{%
    \raisebox{\dimexpr\dimen0-\dimen2}{%
      \scalebox{1}[-1]{\box0}%
    }%
  }%
  {#2}%
}
\makeatother

\makeatletter
    \newcommand{\vast}{\bBigg@{3}}
    \newcommand{\Vast}{\bBigg@{3.5}}
    \newcommand{\vastt}{\bBigg@{4}}
    \newcommand{\Vastt}{\bBigg@{4.5}}
\makeatother

\title[running={Truncated multiplication and batch software implementation for faster mod. exp.}]{Truncated multiplication and batch software SIMD AVX512 implementation for faster Montgomery multiplications and modular exponentiation}

\addauthor[
orcid = {0009-0008-8658-0064},
inst = {1},
email = {laurent-stephane.didier@univ-tln.fr},
surname = {Didier},
]{Laurent-St{\'e}phane Didier}

\addauthor[
orcid = {0000-0003-3840-584X},
inst = {2},
onclick = {https://www.emse.fr/~nadia.el-mrabet/},
email = {nadia.el-mrabet@emse.fr},
surname = {El Mrabet},
]{Nadia El Mrabet}

\addauthor[
orcid = {0009-0008-0966-0503},
inst = {1},
email = {lea.glandus@univ-tln.fr},
surname = {Glandus},
]{L{\'e}a Glandus}

\addauthor[
orcid = {0000-0002-9634-5729},
inst = {1},
onclick = {http://jrobert.univ-tln.fr},
email = {jean-marc.robert@univ-tln.fr},
surname = {Robert},
]{Jean-Marc Robert}

\addaffiliation[ror = {038a20b58},
department = {Laboratoire IMath, Université de Toulon},
postcode = {83130},
country = {France}
]{Toulon}

\addaffiliation[ror = {05a1dws80},
department = {Mines Saint-Etienne, CEA-LETI, 
Centre CMP, Department SAS},
postcode = {42100},
country = {France}
]{Saint-Etienne}

\addfunding[ror=04mdawb03,
            fundref=501100018833,
            country={France}]{Agence de l'innovation de défense}

\genericfootnote{This work has been partially funded by the AID (Agence pour l'Innovation de la Défense), project 2022151.}

\begin{document}

\maketitle

\keywords[Montgomery Modular multiplication, Modular Exponentiation, Batch Computation, AVX512, SIMD, RSA]{Montgomery Modular multiplication, Modular Exponentiation, Batch Computation, \texttt{AVX512}, SIMD, RSA}

\abstract{This paper presents software implementations of batch computations, dealing with multi-precision integer operations. In this work, we use the Single Instruction Multiple Data (SIMD) \texttt{AVX512} instruction set of the \texttt{x86-64} processors, in particular the vectorized fused multiplier-adder \texttt{VPMADD52}. We focus on batch multiplications, squarings, modular multiplications, modular squarings and constant time modular exponentiations of 8 values using a word-slicing storage. We explore the use of Schoolbook and Karatsuba approaches with operands up to 4108 and 4154 bits respectively. We also introduce a truncated multiplication that speeds up the computation of the Montgomery modular reduction in the context of software implementation. Our Truncated Montgomery modular multiplication improvement offers speed gains of almost 20\% over the conventional non-truncated versions. Compared to the state-of-the-art \texttt{GMP} and \texttt{OpenSSL} libraries, our speedup modular operations are more than 4 times faster.  Compared to \texttt{OpenSSL BN\_mod\_exp\_mont\_consttimex2} using \texttt{AVX512} and \texttt{VPMADD52} in 256-bit registers, in fixed-window exponentiations of sizes $1024$ and $2048$, our 512-bit implementation provides speedups of respectively 1.75 and 1.38, while the 256-bit version speedups are 1.51 and 1.05 for $1024$ and $2048$-bit sizes (batch of 4 values in this case).}

\begin{textabstract}
This paper presents software implementations of batch computations, dealing with multi-precision integer operations. In this work, we use the Single Instruction Multiple Data (SIMD) AVX512 instruction set of the x86-64 processors, in particular the vectorized fused multiplier-adder VPMADD52. We focus on batch multiplications, squarings, modular multiplications, modular squarings and constant time modular exponentiations of 8 values using a word-slicing storage. We explore the use of Schoolbook and Karatsuba approaches with operands up to 4108 and 4154 bits respectively. We also introduce a truncated multiplication that speeds up the computation of the Montgomery modular reduction in the context of software implementation. Our Truncated Montgomery modular multiplication improvement offers speed gains of almost 20 \% over the conventional non-truncated versions. Compared to the state-of-the-art GMP and OpenSSL libraries, our speedup modular operations are more than 4 times faster.  Compared to OpenSSL BN\_mod\_exp\_mont\_consttimex2 using AVX512 and madd52* (madd52hi or madd52lo) in 256-bit registers, in fixed-window exponentiations of sizes $1024$ and $2048$, our 512-bit implementation provides speedups of respectively 1.75 and 1.38, while the 256-bit version speedups are 1.51 and 1.05 for $1024$ and $2048$-bit sizes (batch of 4 values in this case).
\end{textabstract}


\section{Introduction}

\normalsize

The need of multi-precision computation in the context of cryptographic operations arose with the wide use of public key cryptography, linked to the RSA cryptosystem~\cite{RivestSA78}, Digital Signature Algorithm (DSA)~\cite{BruceShcheierThese}, Elliptic Curve Cryptography~\cite{ECC99}, or isogeny-based cryptography~\cite{sike}. While these protocols are destined to give way to the future post-quantum ones in the long term, they are still widely used and the need of high throughput signature and/or verification remains vital. One direction to ensure fast signature or verification computations is to make use of the modern processor features. The \texttt{x86-64} processors include a 64-bit multiplier providing a 128-bit result. This instruction is the foundation of several multi-precision libraries like \texttt{GMP} (see~\cite{gnu_mp}). Since the appearance of the SIMD (Single Instruction Multiple Data) instruction sets, especially the \texttt{AVX} and \texttt{AVX2}, allowing parallel computations, several works attempted to exploit the possibilities of these instructions. In 2012, Gueron and Krasnov in~\cite{gueron2012software} proposed the so called \texttt{RSA\_Z} implementation for \texttt{OpenSSL}. These software implementations take advantage of the 32-bit parallel multipliers, i.e. four 32-bit multiplications (providing four 64-bit results) can be performed simultaneously using the \texttt{AVX2} 256-bit registers. However, these improvements rely on the microarchitecture version used, since the complexity balance in terms of 32-bit multiplications leads to the same multiplication instruction numbers in comparison with conventional 64-bit sequential architecture. Gueron \emph{et al.} in~\cite{gueron2016accelerating,NDruckerGK18} proposed other works using \texttt{AVX512} to improve this balance. However, the carry management and word alignment of the additions often necessitates time consuming operations (shuffles or permutes, shifts, additions...) and the efficiency of the computation of a single multi-precision multiplication using SIMD instruction sets is not as efficient as expected. Furthermore, the carry management leads to the so called reduce-radix approach, that is  using 28 or 27-bit operands in 32-bit multipliers, in order to keep spare bits to ensure the carry management, with a penalty in execution speed. Recently, the availability of the \texttt{VPMADD52} instructions, computing up to 8 operations of the form $a+b\times c$ (\texttt{AVX512}) in 64-bit word-slicing with $b$ and $c$ of size 52 bits, brought some improvements. Nevertheless, the speedup remains small, see Gueron \emph{et al.} in~\cite{gueron2016accelerating} (see also Bos \emph{et al.} in~\cite{bmsz13}, Edamatsu and Takahashi in~\cite{tg19,et20} and Takahashi in~\cite{tak20}).

To counteract the carry management and word alignment penalty, another way is to design batch multipliers (and batch computations in general), using a word slicing representation of a corresponding batch of values. This technique is well known and used in the GPU context (see~\cite{bernstein2009billion,trei2013efficient,mahe2014fast,emmart2016optimizing,emmart2018faster,antao2010elliptic,bos2012low}).

In 2008, Grabher \emph{et al.} \cite{GrabherGP2008} present cryptographic pairing software implementations, some of them using a technique close to the word slicing approach with reduced radix (29 bit witdh) for inter-pairing parallel computations.
Quite recently (2022), Buhrow \emph{et al.} in~\cite{buhrow2022parallel} improved the concept with 32-bit SIMD multipliers using the \texttt{AVX512} instruction set (512-bit registers), applied to CRT-RSA decryptions. These implementations compute 8 modular exponentiations simultaneously, using the reduce-radix architecture, providing up to 1.9 speedup in comparison with the \texttt{OpenSSL} throughput, for $2048$-bit CRT-RSA decryptions. With the \texttt{AVX512} and \texttt{VPMADD52} instructions, and in the context of post-quantum SIKE protocol, Cheng \emph{et al.} in~\cite{cheng2021batching,cheng2022highly} proposed a batch implementation of the key exchange protocol. However, their implementations are specific to the SIKE parameters, in particular the Montgomery friendly primes for the modular operations. They make use of a slightly reduced radix approach (51-bit operands for the 52-bit multiplier of the \texttt{VPMADD52}). To the best of our knowledge, these multipliers are the fastest for the considered sizes (e.g. 503 bits), but with a relaxed carry management due to the specific modular reduction following each multiplication or squaring. Therefore, the implementations of Cheng \emph{et al.} cannot be used or transposed in other contexts.

Recently, the \texttt{OpenSSL} library offers a new approach, which is intermediary between a parallelized computation of a single exponentiation and a batch version, implemented by Kyrillov and Matyukov (see\cite{openssl}). This implementation computes two exponentiations simultaneously using the \texttt{VPMADD52} instruction on 256-bit registers (when available on the platform). This avoids some microarchitecture issues of \texttt{AVX512} while it is compliant to the Intel \texttt{AVX10.1} new extension (see~\cite{intelavx10}). This function processes operands of 1024, 1536 and 2048 bits. To the best of our knowledge, these implementations are the state-of-the-art.

Another topic we have focused on in this paper is the concept of truncated operation applied to modular reduction. In his seminal paper, Barrett in~\cite{barrett} mentioned the idea, but did not implement the concept further. Montgomery in~\cite{montMult_85} did not pay attention to the idea either. Subsequent works on the Montgomery modular reduction or multiplication developed the concept of block or word approaches, like Koç in~\cite{kocAT96}. These are the CIOS (for Coarsely Integrated Operand Scanning) and variants, which aim to improve the performance by optimising the word addition numbers and the memory access patterns. However, these approaches are not suitable for truncated operations. Later on, Hars in~\cite{Hars2005} and~\cite{Hars2006} studied the concept of truncated multiplications and their application to various contexts, mainly Barrett and Montgomery multiplications. This work remained theoretical and focused on possible improvements in hardware implementations but no implementations have been developed. More recently, Ding \emph{et al.} in~\cite{DingS2018} proposed a hardware implementation of a 256-bit ECC processor using a Barrett modular multiplication with a truncated multiplication. Later on in~\cite{DingS2020}, the same authors proposed FPGA implementations of 256 and 512-bit Montgomery modular multipliers based on 3 and 4 way Karatsuba truncated multipliers. However, these approaches cannot be easily transferable to the context of software implementations. Furthermore, since they only present very few versions and sizes, their work is difficult to generalize.

In another recent work, Bos \emph{et al.} in~\cite{BosKP2021} explore parallel implementations of Montgomery Multiplications in both cases: intra-parallelism of CIOS approach with a multi-thread implementation, 
and inter-multiplication parallelism, using a word-slicing representation. However, they do not provide implementation results.

\paragraph*{Contributions.} In this paper, we present software implementations of batch multi-precision multipliers, that is 8 simultaneous multiplications, using a word slicing representation in radix $2^{52}$ to take advantage of the \texttt{AVX512 VPMADD52} instructions. We implement sizes up to $4108$ bits, using the Schoolbook approach and up to $4154$ bits using the Karatsuba approach. We use these multipliers in order to implement Montgomery modular multiplications, and propose, to the best of our knowledge, the first software implementation of Truncated Montgomery modular multiplication, which allows up to 20~\% speedup over non-truncated versions and presents a more than 4 times speedup over the \texttt{OpenSSL RSA\_Z} (not using the 256-bit \texttt{VPMADD52}) and \texttt{GMP} libraries. We make use of these modular multiplications in fixed-window modular exponentiations, for sizes $1024$, $2048$ and $4096$ bits, with speedups up to nearly 4 ($1024$ bits) over the \texttt{OpenSSL BN\_mod\_exp\_mont\_consttime} and 1.75 over the \texttt{OpenSSL BN\_mod\_exp\_mont\_consttimex2} which computes two exponentiations in parallel using the 256-bit \texttt{VPMADD52}. We also implemented a 256-bit Truncated Montgomery exponentiation of our batch \texttt{AVX512} using \texttt{VPMADD52} instructions one in order to compare with the corresponding \texttt{OpenSSL BN\_mod\_exp\_mont\_consttimex2}, with a maximum speedup of 1.51 and 1.05 for 1024 and 2048-bit operands respectively.

\paragraph*{Organisation of the paper.} This paper is organised as follows: Section~\ref{sec:batch} presents the implementation principles of the Schoolbook batch multiplications and squarings, Section~\ref{sec:batchK} deals with the Karatsuba versions of the batch multiplications and squarings, Section~\ref{sec:batchMM} presents the Montgomery modular batch multiplications. In Section~\ref{sec:batchTruncMM}, we then present the Truncated Montgomery modular reduction and its batch implementation. This is followed by the presentations of the implementation performance in Section~\ref{sec:batchPerfs}, including the fixed-window modular exponentiations. A conclusion ends the paper.

\paragraph{Notations.}

\begin{itemize}
    \item $>>$ represents a logic RIGHT SHIFT
    \item  $\bigvee$ represents a logic OR
    \item  $\&$ represents a logic AND
\end{itemize}

Multiprecision numbers will be represented by either bits denoted as $a_{k_i}$, or in 64 bit word arrays denoted as $A64_k[i]$, or in 52 bit word arrays denoted as $A_k[i]$.


\section{Batch Schoolbook multiplications}
\label{sec:batch}

The objective of our implementations of batch Schoolbook multiplications is to achieve eight simultaneous multiplications when using the \texttt{AVX512} instruction set, or a batch of four values when using the \texttt{AVX2} one and will further refer to their corresponding C variables types as \mavxi and \texttt{\_\_m256i} which are respectively 256 and 512 bit vectors, representing integers. In the sequel, we focus on the \mavxi case since the conversion to the other case is straightforward. We use a word slicing approach of the operands as follows.

Let us have a batch of eight values $A_k, 0\leq k <8$ of $t$~bits each. They can be sliced into $t_{64} = \lceil t/64\rceil$ words of 64 bits each.  To be able to use the \texttt{VPMADD52} instructions, we propose to split them into 52-bit words. This requires the use of $t_{52} = \lceil t/52\rceil$ 512-bit lines to store the 8 $A_k$'s.

More formally, we set: $$A_k = \sum^{t}_{i=0} a_{ki} 2^{i}=\sum^{t_{64}}_{i=0} A64_{k}[i] 2^{64\times i} = \sum^{t_{52}}_{i=0} A_{k}[i] 2^{52\times i},$$

and store the eight values as shown in Table~\ref{tab:wordslicing}.

\medskip
\begin{table*}[ht]
\caption{Word Slicing storage of eight values in 52-bit words}
\begin{center}
\scriptsize
\begin{tabular}{c|c|c|c|c|c|c|c|c|}
 &	$\overbrace{~~~~~~~~~~~~}^{\mbox{64 bits}}$	&	$\overbrace{~~~~~~~~~~~~}^{\mbox{64 bits}}$&	$\overbrace{~~~~~~~~~~~~}^{\mbox{64 bits}}$&	$\overbrace{~~~~~~~~~~~~}^{\mbox{64 bits}}$&	$\overbrace{~~~~~~~~~~~~}^{\mbox{64 bits}}$&	$\overbrace{~~~~~~~~~~~~}^{\mbox{64 bits}}$&	$\overbrace{~~~~~~~~~~~~}^{\mbox{64 bits}}$&	$\overbrace{~~~~~~~~~~~}^{\mbox{64 bits}}$	\\
&	$~~~\overbrace{~~~~~~}^{\mbox{52 bits}}$	&	$~~~\overbrace{~~~~~~}^{\mbox{52 bits}}$	&	$~~~\overbrace{~~~~~~}^{\mbox{52 bits}}$	&$~~~\overbrace{~~~~~~}^{\mbox{52 bits}}$	&$~~~\overbrace{~~~~~~}^{\mbox{52 bits}}$	&$~~~\overbrace{~~~~~~}^{\mbox{52 bits}}$	&$~~~\overbrace{~~~~~~}^{\mbox{52 bits}}$	&	$~~~\overbrace{~~~~~}^{\mbox{52 bits}}$	\\

\hline
  $A512_0$	&	$0~|~A_{0}[0]$	&	$0~|~A_{1}[0]$	&	$0~|~A_{2}[0]$	&	$0~|~A_{3}[0]$	&	$0~|~A_{4}[0]$	&	$0~|~A_{5}[0]$	&	$0~|~A_{6}[0]$	&	$0~|~A_{7}[0]$	\\
\hline
 $A512_1$		&	$0~|~A_{0}[1]$	&	$0~|~A_{1}[1]$	&	$0~|~A_{2}[1]$	&	$0~|~A_{3}[1]$	&	$0~|~A_{4}[1]$	&	$0~|~A_{5}[1]$	&	$0~|~A_{6}[1]$	&	$0~|~A_{7}[1]$	\\

\hline
 \vdots &\multicolumn{8}{c|}{\vdots}\\
 \vdots &\multicolumn{8}{c|}{$t_{52}$ 512-bit lines}\\
 
\hline

\end{tabular}
\label{tab:wordslicing}

\end{center}
\end{table*}

\normalsize

We present the \texttt{VPMADD52} instructions, which are integer fused multiply–add instructions, and denote them using the C-intrinsics \cite{intelintrinsics}:

\noindent\texttt{\_mm512\_madd52lo\_epu64(a, b, c)}  

\noindent\texttt{\_mm512\_madd52hi\_epu64(a, b, c)} 

\small
\begin{enumerate}
    \item \texttt{\_\_m512i \_mm512\_madd52lo\_epu64(a, b, c)}
\begin{algorithmic}
\For{$i=0 \dots 7$}
\State // 52 least significant bits of $b\times c$
\State $Dest_i \gets a_i+[(b_i\times c_i) \bmod 2^{52}]$ \EndFor
\end{algorithmic}
    
    \item
    \texttt{\_\_m512i \_mm512\_madd52hi\_epu64(a, b, c)}
    
    \begin{algorithmic}
    \For {$i=0 \dots 7$}
    \State // 52 most significant bits of $b\times c$
    \State $Dest_i \gets a_i+[(b_i\times c_i) >> 52]$ 
    \EndFor
    \end{algorithmic}
\end{enumerate}
\normalsize

In the rest of the paper we rename these instructions to \texttt{madd52lo} and \texttt{madd52hi}.

\subsection{Schoolbook multiplication}\label{sub:school}

We target operands of RSA sizes, and more generally sizes which are multiples of 64. Because we also aim at the \texttt{VPMADD52} instructions, we convert radix $2^{64}$ representation into radix $2^{52}$ representation. This minimizes as much as possible the number of words required to represent the operands. 

For instance, in~\cite{cheng2022highly}, the authors use radix $2^{51}$ which is convenient for the SIKE 503 case they focused on. They use 10-word operands with Schoolbook or Karatsuba approaches. However, this requires 11 words in case of 512-bit operands whereas 10 words are enough in radix $2^{52}$.

After testing several configurations, we choose the Algorithm~\ref{alg:8SBmul} (\texttt{B\_mul}) in order to limit the memory transfer of operands. Thus, we use the same 512-bit line of the first operand, parsing the second operand lines and storing the partial products in the corresponding line of the result. This configuration necessitates a carry management at the end which is performed using a binary mask \textsl{mask52}.

\begin{algorithm}[ht]
  \caption{Batch Schoolbook multiplications, \texttt{B\_mul}}
  \label{alg:8SBmul}
  \small
  \begin{algorithmic}[1]
    \Require  Two batches of 8 values $A_k$ and $B_k$ stored in 52-bit slices in $t_{52}$ \mavxi shares, $mask52$ is a 512-bit batch of eight 52-bit masks.
    \Ensure A batch of 8 values $C_k = A_k\times B_k$ stored in 52-bit slices in $2\times t_{52}$ \mavxi shares.
    \For{$k$ from 0 to 7} {\bf in parallel}
    \State $C_k[0] \gets 0$
        \For{$i$ from $0$ to  $t_{52}-1$}
    	    \For{$j$ from $0$ to  $t_{52}-1$}
    		    \State $C_k[i+j] \gets \textrm{\texttt{madd52lo}}(C_k[i+j],\; A_k[i],\; B_k[j])$
    	    \EndFor
	    \EndFor

    \State $carry_k\gets 0$ \Comment{carry management}
    \For{$i$ from $2$ to  $2\times t_{52}-2$}
		\State  $carry_k \gets  C_k[i-1]>>52$
		\State  $C_k[i] \gets C_k[i] + carry_k$
		\State $C_k[i-1]\gets C_k[i-1] \& mask52$
	\EndFor
	\State  $carry_k \gets  C_k[2\times t_{52}-2]>>52$
	\State  $C_k[2\times t_{52}-1] \gets C_k[2\times t_{52}-1] + carry_k $

      \EndFor
    \State {\bf return} $C$
  \end{algorithmic}
\end{algorithm}

\subsection{Squaring}

We applied a similar approach for the squaring operation (Alg.~\ref{alg:8SBsqu}, \texttt{B\_square}).
In this algorithm, the number of \texttt{VPMADD52} operations is divided by nearly two compared to the multiplication implementation presented in   Section~\ref{sub:school}.

\begin{algorithm}[ht]
  \caption{Batch Schoolbook squaring \texttt{B\_square}}
  \label{alg:8SBsqu}
\small
  \begin{algorithmic}[1]
    \Require  One batch of 8 values $A_k$ stored in 52-bit slices in $t_{52}$ \mavxi shares, $mask52$ is a 512-bit batch of eight 52-bit masks.
    \Ensure A batch of 8 values $C_k = A_k^2$ stored in 52-bit slices in $2\times t_{52}$ \mavxi shares.
    \State $carry\gets 0_{512 bits}$
    \For{$k$ from 0 to 7} {\bf in parallel}
        \For{$\ell$ from $0$ to  $2\times t_{52}-1$}
            \State $C_k[\ell]\gets 0_{512}$
            \For{$(i,j)$ such that $i+j=\ell$ and $j<i$}
                \State $C_k[\ell]\gets \textrm{\texttt{madd52lo}}(C_k[\ell],A_k[i],A_k[j])$
            \EndFor
            \For{$(i+j)$ such that $i+j=\ell-1$ and $j<i$}
                \State $C_k[\ell]\gets \textrm{\texttt{madd52hi}}(C_k[\ell],A_k[i],A_k[j])$
            \EndFor
            \State $C_k[\ell] \gets C_k[\ell]<<1$
            \If{$\ell \bmod 2=0$}
                \State $i\gets \ell/2$
                \State $C_k[\ell]\gets \textrm{\texttt{madd52lo}}(C_k[\ell],A_k[i],A_k[i])$
            \Else
                \State $i\gets \lfloor \ell/2\rfloor$
                \State $C_k[\ell]\gets \textrm{\texttt{madd52hi}}(C_k[\ell],A_k[i],A_k[i])$
            \EndIf
            \State $C_k[\ell] \gets C_k[\ell] + carry_k $
            \State $carry_k\gets C_k[\ell]>>52$
            \State $C_k[\ell]\gets C_k[\ell] \& mask52$
        \EndFor
    \EndFor
    \State return $C$
  \end{algorithmic}
\end{algorithm}

\subsection{Complexity comparison}
As mentioned above, for $t$-bit operands split in 64-bit words, we store 8 operands in $t_{52} = \lceil t/52\rceil$ 512-bit lines. The Schoolbook batch multiplication requires $t_{52}^2$ elementary multiplications, performed using both instructions (\texttt{madd52lo} and \texttt{madd52hi}), that is $2\times t_{52}^2$ \texttt{VPMADD52} instructions. Algorithm~\ref{alg:8SBmul} requires a few extra 512-bit additions in order to manage the carries.

    The batch squaring requires only $t_{52} (t_{52}+1)$ \texttt{VPMADD52} instructions and $t_{52}-1$ left shifts (i.e. multiplications by 2) and additions also in order to handle the carries.
These complexities are summarized in Table~\ref{tab:compSBmul}.

\begin{table*}[b]
  \caption{Number of instructions of batch Schoolbook multiplications}
  \label{tab:compSBmul}
  \begin{center}
  \small
  \begin{tabular}{|c|c|c|c|c|}
  	\hline
  			&	\# \texttt{VPMADD52}	&	\# shifts	&	\# Additions	& \# maskings	\\
	\hline
	Multiplication Alg.~\ref{alg:8SBmul}	&	$2\times t_{52}^2$			&	$2\times t_{52}-2$	&	$2\times t_{52}-2$	&	$2\times t_{52}-2$\\
	\hline
	Squaring Alg.~\ref{alg:8SBsqu}			&	$t_{52}\times (t_{52}+1)$	&	$2\times t_{52}-2$	&	$2\times t_{52}-2$	&	$2\times t_{52}-2$\\
	\hline
  \end{tabular}
  \end{center}
\end{table*}


\section{Batch Karatsuba multiplications}
\label{sec:batchK}

This batch construction is also adapted to Karatsuba multiplication. We remind here this construction.
Let $A$ and $B$ be two $t$-bit operands. We assume that $t$ is even. We first split the operands as follows:
$$A = a_\ell + 2^{t/2}a_h, B = b_\ell + 2^{t/2}b_h.$$
We then compute 3 elementary products (instead of four in the Schoolbook method) and we similarly split them:
$$\left\{\begin{array}{cl}
    D_0 = D_{0\ell} + 2^{t/2}D_{0h} \gets& a_\ell\times b_l,\\
    D_1 = D_{1\ell} + 2^{t/2}D_{1h} \gets& (a_\ell+a_h)\times (b_\ell+b_h),\\
    D_2 = D_{2\ell} + 2^{t/2}D_{2h} \gets& a_h\times b_h.
\end{array} \right.$$

The low-level multiplications are performed with the Schoolbook method. The multiplication width is determined by $(a_\ell+a_h)\times (b_\ell+b_h)$. Finally, we obtain the result as follows:

$$\begin{array}{cl}A\times B =& D_{0\ell}\\
    &+ 2^{t/2}(D_{0h}+D_{1\ell} - D_{0\ell}-D_{2\ell})\\
    &+ 2^t(D_{2\ell}+D_{1h} - D_{0h}-D_{2h})\\
    &+ 2^{3t/2}D_{2h}.
\end{array}$$

The size of $A$ and $B$ and the size of the elementary multiplication are linked. For example, if the largest elementary product requires 520-bit operands, $A$ and $B$ have to be at most $2\times$ 519 bits long. This means that  $t=$1038.
The Table~\ref{tab:operand_size} sums up the sizes we consider in the rest of the paper.

\begin{table}[hb]
     \caption{Operand and elementary multiplication size}
    \label{tab:operand_size}
   \centering
    \small
\begin{tabular}{|c|c|c|c|c|}
    \hline
     Karatsuba mult. size $t$&  518 &   1038    &   2078    &   4154$^*$\\
     \hline
     Elementary product&    260 &   520 &   1040    &   2078 \\
     \hline
     \multicolumn{5}{l}{$^*$ double Karatsuba}
\end{tabular}
\end{table}

In this construction, we represent the operands in two radix $2^{t/2}$ shares. They are radix $2^{52}$ numbers in case of one Karatsuba stage. In the special case of 4154 bits which requires two Karatsuba stages, we have the following representation:

$$A = a_\ell + 2^{2077} a_h,$$
with $a_\ell = a_{\ell\ell} + 2^{1039}a_{\ell h}$ and $a_h = a_{h\ell} + 2^{1039}a_{hh}$. The shares are then represented with 52-bit words. 
All these representations are stored in memory in the batch way presented Section~\ref{sec:batch}.

In terms of complexity, the number of Schoolbook elementary multiplications is 3 for 1038 and 2078 bits (respectively 520 and 1040-bit multiplications), and 9 Schoolbook elementary multiplications of size 1040 bits for the 4154-bit batch Karatsuba multiplications. We provide the instruction count in Table~\ref{tab:compKmul}.

While the \texttt{VPMADD52} instruction count is smaller compared to the Schoolbook case, the other instructions are much more numerous. This explains why the Karatsuba approach is interesting only for the 4154-bit case for the squaring.

\begin{table*}
  \caption{Number of instructions of batch Karatsuba multiplications}
  \label{tab:compKmul}
  \begin{center}
  \small
  \begin{tabular}{|c|c|c|c|c|}
  	\hline
  			&	\# \texttt{VPMADD52}	&	\# shifts	&	\# Add./Sub.	& \# maskings	\\
	\hline
	\begin{tabular}{c}Multiplication\\ with 3 elt. Alg.~\ref{alg:8SBmul}
    \end{tabular}
	&	$\frac 3 2\times t_{52}^2$			&	$8\times t_{52}-6$	&	$10\times t_{52}-7$	&	$8\times t_{52}-6$\\
	\hline
	\begin{tabular}{c}Squaring\\ with 3 elt. Alg.~\ref{alg:8SBsqu}
    \end{tabular}		&	$\frac 3 4 t_{52}\times (t_{52}+2)$	&	$10\times t_{52}-1$	&	$9\times t_{52}+3$	&	$7\times t_{52}+4$\\
	\hline
 
  \end{tabular}
  \end{center}
\end{table*}


\section{Batch Montgomery modular multiplications}
\label{sec:batchMM}

Let us briefly remind of the Montgomery modular multiplication~\cite{montMult_85}. In this operation, the product $T$ of two $t$-bit operands is reduced modulo a $t$-bit modulus $N$ (Alg.~\ref{alg:montred}). The returned result is $C\gets T \times R^{-1}\bmod N$ where $R$ is a power of 2 in binary implementations. In order to handle this multiplicative factor, the input operands are usually converted in the Montgomery representation. This consists of multiplying the initial operands by the square modulo $R$ of $N$, using the same Montgomery modular reduction. By this, one gets:

$$C \gets MontRed(T\times R^2, N) = \frac{A\times R^2} R \mod N$$
and 
$$C = T \times R \bmod N.$$

This renders the representation stable in case of a sequence of multiple multiplications and squarings.
A final Montgomery reduction is enough to convert the result from this Montgomery representation.

\paragraph{Batch Montgomery multiplication.}
With this context, it is possible to make a batch Montgomery reduction modulo 8 different moduli $N_i$ using \texttt{AVX512} instructions. We investigated several approaches.

The first is based on the Schoolbook multiplication and/or squaring (Alg.~\ref{alg:8SBmul} and Alg.~\ref{alg:8SBsqu}) followed by the batch Montgomery reduction (Alg.~\ref{alg:8MontRed}). In this case, we use a batch fused multiplier-adder (\texttt{B\_fma}) operation at line 2 of Alg.~\ref{alg:8MontRed}. This operation has the same cost as a single multiplication because of the use of the \texttt{VPMADD52} instructions. This makes free the addition required in the Montgomery reduction. 

With the Karatsuba multiplication, it is not possible to use a batch Karatsuba \texttt{B\_fma} because of the additions in the reconstruction phase. This explains why the speedups are slightly better for the small sizes in the Schoolbook case.

\sloppypar{We also implemented the word-level Montgomery reduction, adapted from~\cite{kocAT96} also called CIOS and variants (BPS improvement in~\cite{buhrow2022parallel}), only on the Montgomery multiplication. This approach is presented Algorithm~\ref{alg:8BlockMontRed}.}

In both Algorithms~\ref{alg:8MontRed} and~\ref{alg:8BlockMontRed}, the usual final subtraction is not required in our case, since the size of the multiplication is greater enough than the size of the modulus, see Gueron \emph{et al.} in~\cite{gueron2012software}.

\begin{algorithm}[ht]
  \caption{Batch Montgomery reduction}
  \label{alg:8MontRed}
  \small
  \begin{algorithmic}[1]
    \Require  One batch of 8 values $A$ stored in 52 word-slices in $2\times t_{52}$ \mavxi shares, $mask52$ is a 512-bit batch of eight 52-bit masks, the 8 moduli $N$, some precomputed values $N' = (-N)^{-1} \bmod R$ with $R = 2^{52\times t_{52}}$.
    \Ensure A batch of 8 values such that $C_k = A_k\times R^{-1}  \bmod N_k$ stored in 52 word-slices in $t_{52}$ \mavxi shares.
    \State $q\gets \mathtt{B\_mul}(A,N') \bmod R$
    \State $T \gets \mathtt{B\_fma}(q,N,A)$\Comment{$A+q\times N$}
    \State $C \gets T/R$ \Comment{returns the $t$ higher words of $T$}
    \State return $C$
  \end{algorithmic}
\end{algorithm}

\begin{algorithm}[ht]
  \caption{Batch CIOS inspired Montgomery multiplication}
  \label{alg:8BlockMontRed}
  \small
  \begin{algorithmic}[1]
    \Require  Two batches of 8 values $A_k$ and $B_k$ stored in 52 word-slices in $2\times t_{52}$ \mavxi shares, $mask52$ is a 512-bit batch of eight 52-bit masks, the 8 moduli $N_k$, some precomputed values $N'_k = (-N_k)^{-1} \bmod R$ with $R = 2^{52\times t_{52}}$.
    \Ensure A batch of 8 values such that $C_k = A_k\times B_k\times R^{-1}  \bmod N$ stored in 52 word-slices in $t_{52}$ \mavxi shares.
    \For{$k$ from 0 to 7} {\bf in parallel}
        \State $Y_k \gets a_k[0]\cdot B_k$
        \State $q_k \gets |Y_k|_{2^{52}}\cdot N'_k \mod 2^{52}$
        \State $Y_k \gets (Y_k  + q_k\cdot N_k)/2^{52}$ 
        \For{$i=1$ \textbf{to} $t_{52}-1$}
            \State $Y_k \gets Y_k+a_k[i]\cdot B_k$
            \State $q_k \gets |Y_k|_{2^{52}}\cdot N'_k \mod 2^{52}$
            \State $Y_k \gets (Y_k  + q_k\cdot N_k)/2^{52}$ 
        \EndFor
    \EndFor
    \State return $C \gets Y$  \end{algorithmic}
\end{algorithm}


\section{Truncated Batch Montgomery modular multiplications}
\label{sec:batchTruncMM}

The classic Montgomery reduction is reminded in Algorithm~\ref{alg:montred}. This algorithm computes $T \times R^{-1}\bmod N$ where the modulus $N$ is a  $t$-bit integer (with $t \equiv 0 \bmod 64$) 
and the Montgomery constant $R = 2^{52\times t_{52}}$.

\begin{algorithm}[ht]
  \caption{Montgomery modular reduction: \emph{MontRed}}
  \label{alg:montred}
  \small
  \begin{algorithmic}[1]
    \Require  $T < 4N^2$, $N$ the $t$-bit modulus, 
    $R = 2^{52\times t_{52}}$, with $t_{52}\geq t$, precomputed value $N' = (-N)^{-1} \bmod R$.
    \Ensure $C \equiv T \times R^{-1}\bmod N$ and $C<2N$
    \State $q\gets T\times N' \bmod R$ \Comment{$q<R$}
    \State $C \gets \frac {T+q\times N}R$ \Comment{$C<2N$}
    \State \Return $C$
  \end{algorithmic}
\end{algorithm}

At step 2 of Algorithm~\ref{alg:montred}, the division by $R$ is exact because $T+q\times N \equiv 0 \bmod R$. In other words, the $52\times t_{52}$ least significant bits of $T+q\times N$ are all zeroes. 
Because we know the value of its least significant part, the computation of the $52\times t_{52}$ least significant bits of $q\times N$ can be avoided. We only need to estimate the input carry in the sum of the $52\times t_{52}$ most significant bits. These remarks lead to Algorithm~\ref{alg:trunc_montred}.

In this algorithm, $T$ is a $2t_{52}$-word integer and the modulus $N$ is a $t_{52}$-word integer. We denote $T[0]$ the least significant 52-bit word of $T$. At step 2, we compute only the most significant bits of the multiplication $q\times N$, denoted $\widehat{qN} \gets \lfloor \frac{q\times N}R \rfloor$. 
The carry $c_{add}$ to be propagated from the lower part is computed at line 3. Finally, the result is computed at line 4, where $\lfloor T/R \rfloor$ is the $t_{52}$ most significant words of $T$.

\begin{algorithm}[ht]
  \caption{Montgomery modular reduction with truncated multiplication}
  \label{alg:trunc_montred}
  \small
  \begin{algorithmic}[1]
    \Require  $T < 4N^2$, $N$ the $t$-bit modulus, $R = 2^{52\times t_{52}}$, with $t_{52}\geq t$, precomputed value $N' = (-N)^{-1} \bmod R$.
    \Ensure $C \equiv T \times R^{-1}\bmod N$ and $C<2N$
    \State $q\gets T\times N' \bmod R$\Comment{$q<R$}
    \State $\widehat{qN} \gets \lfloor (q\times N)/R \rfloor$
    \State $c_{add} = \mbox{$\bigvee$}_{i=0}^{t_{52}-1} T[i]$
    \State $C \gets\lfloor T/R  \rfloor+ \widehat{qN} + c_{add}$
    \State \Return $C$
  \end{algorithmic}
\end{algorithm}

In the rest of this paper, we deal with batch SIMD software implementations, however, the idea may be applied to classical sequential implementations.

Theorem~\ref{th:correctness} provides the correctness of our Montgomery modular reduction with truncated multiplication. Theorem~\ref{th:qn} describes how to efficiently compute the correct $\widehat{qN}$.
 
\begin{theorem}[Correctness of the Truncated MontMul]
\label{th:correctness}
With $T < 4N^2$, a $t$-bit modulus $N$, an integer $R = 2^{52\times t_{52}}$, with $t_{52}\geq \lceil t/64\rceil$, precomputed value $N' = (-N)^{-1} \bmod R$, Algorithm~\ref{alg:trunc_montred} correctly computes $C \equiv T \times R^{-1}\bmod N$ and $C<2N$. 
\end{theorem}
\begin{proof}
The key point of Algorithm~\ref{alg:trunc_montred} is that we know that $T+q\times N \bmod R \equiv 0$. In other words, the $52\times t_{52}$ least significant bits are zeroes. This makes easier the computation of the carry $c_{add}$. 

If the least significant bit of $T$ is $0$, then the least significant bit of $q\times N$ is also $0$, and no carry is propagated to the next position. This property holds until the first bit of $T$ equals~$1$.

If the $i^{th}$ bit of $T$ is 1, then the $i^{th}$ bit of $qN$ must also be 1 and a carry must be propagated to position $i+1$ in order to ensure that the $i^{th}$ bit of $T+q\times N$ is 0. Next, to ensure that the next bit of $T+q\times N$ is 0, only one of the $i+1^{th}$ bits of $T$ and $qN$ must be 1. These conditions are summarized in the table below:

\begin{center}
    \begin{tabular}{c||c||c|c||}
   bits    &   $i$   &   \multicolumn{2}{c||}{$i+1$}  \\
    $T$ &  1 &   0  &   1   \\
    $q\times N$  &   1   &   1  &   0   \\
    \hline
    generated carry $c_{add}$   &   1    &   \multicolumn{2}{c||}1    \\
\end{tabular}
\end{center}

Thus, if a carry is generated at position $i^{th}$, then another carry is generated at position $i+1^{th}$. As a consequence, the carry generated at position $52\times t_{52} -1$ is 1 only if there is at least one of the $52\times t_{52} -1$ least significant bits that is 1. This carry is output from the $t_{52}-1^{th}$ word.
\begin{equation}\label{eq:c_add}
    c_{add}=c_{add~52\times (t_{52} -1)} = \bigvee_{i=0}^{t_{52}-1} T[i].
\end{equation}

Finally, in Algorithm~\ref{alg:montred}, the result $C$ is computed as follows:
$$C=\left(\sum_{i=0}^{2\times{t_{52}}-1} T[i]2^{i52} + \widehat{qN}2^{52\times t_{52}} +\uwidehat{qN} \right) >>(52\times t_{52}),$$
where $\uwidehat{qN} = qN \bmod R$. As a consequence:
{\small$$C=\!\!\!\!\sum_{i=t_{52}}^{2 \times{t_{52}} -1} \!\!\! T[i] 2^{52(i-t_{52})} + \widehat{qN} + \left( \sum_{i=0}^{t_{52}-1} \!\! T[i]2^{i52} + \uwidehat{qN} \right) \!>>(52\times t_{52}).$$}

The last term of this sum is $c_{add}$ that we compute with equation (\ref{eq:c_add}) at line 3 in Algorithm~\ref{alg:trunc_montred}.
\end{proof}

In Algorithm~\ref{alg:trunc_montred} the computation of $\widehat{qN}$ do not require to compute all the partial products of $q\times N$, as shown in Theorem~\ref{th:qn}. 

\begin{theorem}[Computation of $\widehat{qN}$ ]\label{th:qn}
The correct computation of $\widehat{qN}$ requires only the partial products of $q\times N$ of weight at least $t_{52}-1$.    
\end{theorem}
\begin{proof}
    
Since we know that $T+q\times N \bmod R \equiv 0$, the computation $\widehat{qN}$ can be simplified. Figure~\ref{fig:qN} illustrates the computation of  $q\times N$. It is not necessary to compute all the least significant partial products of $q\times N$ and sum them in order to know which carries to propagate and to compute $\widehat{qN}$ correctly.

\begin{figure}[ht]
    \centering
\includegraphics[width=0.8\columnwidth]{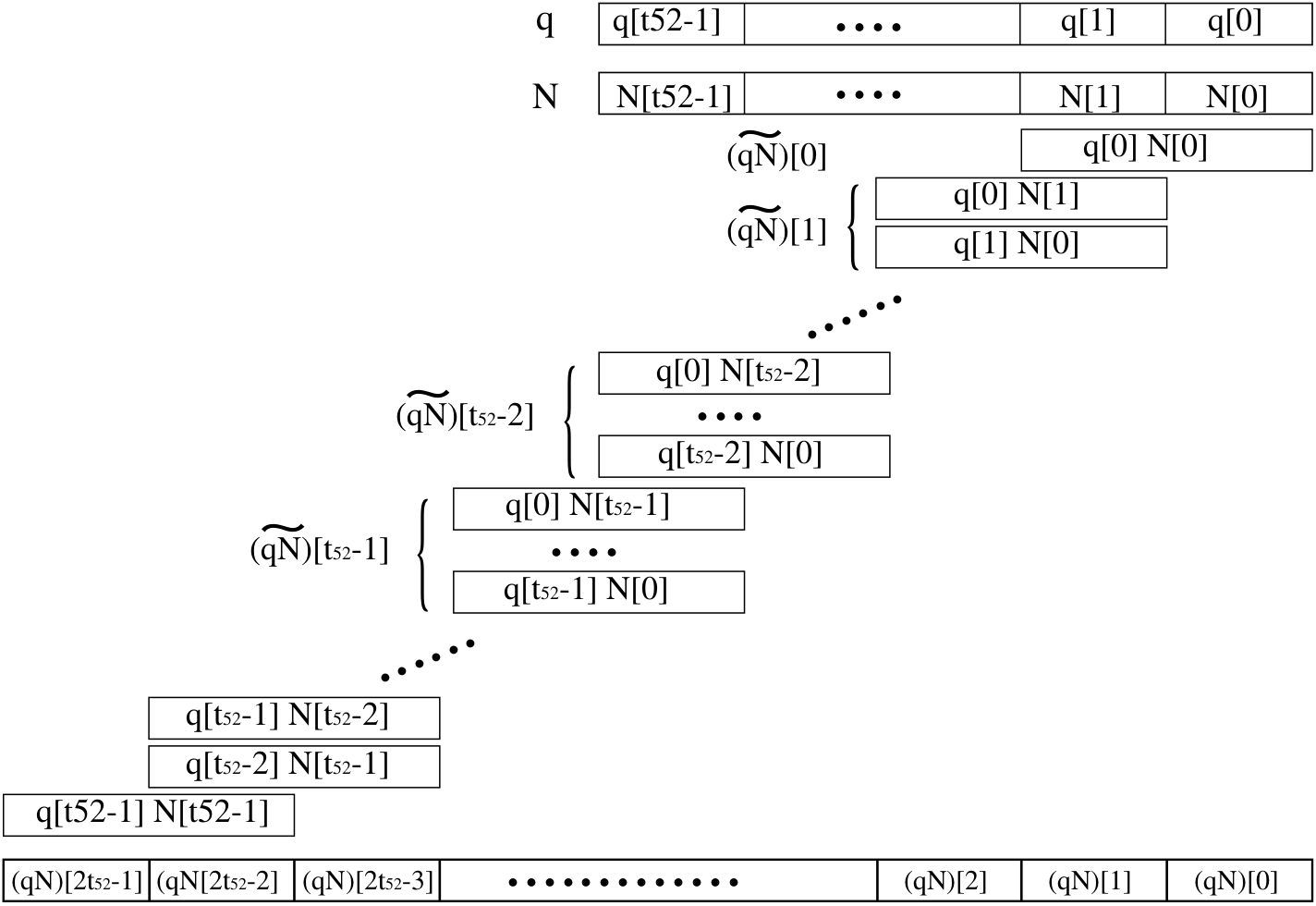}
    \caption{Detail of $q\times N$}
    \label{fig:qN}
\end{figure}

Let us denote $\widetilde{(qN)}[i]$ the sum of the partial products of weight $i$. More formally, the partial product of weight $52t_{52}-1$ is a $\lceil \log_2(t_{52}-1) \rceil + 52$-bit word and is computed as follows: 
\begin{equation}
\label{eq:qn}
	\begin{split}
    \widetilde{(qN)}[t_{52}-1]  \gets & \! \sum^{t_{52}-1}_{i=0,j=0, i+j=t_{52}-1} \!\!\!\!\! \textrm{\texttt{mul52lo}}(q[i],N[j])  \\
     &+ \!\!\!\!\!\! \sum^{t_{52}-2}_{i=0,j=0, i+j=t_{52}-2} \!\!\!\!\! \textrm{\texttt{mul52hi}}(q[i],N[j])\\
     &+ \widetilde{(qN)}[t_{52}-2]>>52,
	\end{split}
\end{equation}
where \texttt{mul52hi} and \texttt{mul52lo} compute respectively the 52 higher and lower bits of two 52-bit operands.
Therefore, the $t_{52}-1^{th}$ word of $T+qN$ (which is known to be 0) is:
$$(T+qN)[t_{52}-1]=(T[t_{52}-1]+ \widetilde{(qN)}[t_{52}-1] + c_{add} ) \bmod 2^{52}$$
and then
\begin{equation}\label{eq:TqN}
    (T[t_{52}-1]+ \widetilde{(qN)}[t_{52}-1] + c_{add} ) \bmod 2^{52} = 0.
\end{equation}

\textbf{If $\mathbf{c_{add}=0}$:} This means that all $T[i]=0$  for $i\leq t_{52}-1$ and so for $q[i]$, because $q=T\times N' \bmod R$ in Alg.~\ref{alg:trunc_montred}. Therefore in equation (\ref{eq:qn}), only the partial products of weight greater than or equal to $t_{52}-1$ are needed.

\textbf{If $\mathbf{c_{add}=1}$:} The only way the equation (\ref{eq:TqN}) can be verified is if the binary vector $T[t_{52}-1]+ \widetilde{(qN)}[t_{52}-1]$ has all its 52 least significant bits set to 1. 

Using Eq. (\ref{eq:qn}), equation (\ref{eq:TqN}) can be written as follows:
\begin{equation}
     (UP + \widetilde{(qN)}[t_{52}-2] >> 52+ c_{add} ) \bmod 2^{52} = 0,
\end{equation}
where 
\[
	\begin{split}
    UP  \gets & T[t_{52}-1] + \!\!\! \sum^{t_{52}-1}_{i=0,j=0, i+j=t_{52}-1} \!\!\!\!\! \!\!\!\!\textrm{\texttt{mul52lo}} (q[i],N[j])  \\
     & ~~~~~~~~~~~~ + \!\!\! \sum^{t_{52}-2}_{i=0,j=0, i+j=t_{52}-2} \!\!\!\!\! \!\! \!\!\textrm{\texttt{mul52hi}}(q[i],N[j]).
	\end{split}
\]

\begin{figure}[ht]
    \centering
\includegraphics[width=0.8\columnwidth]{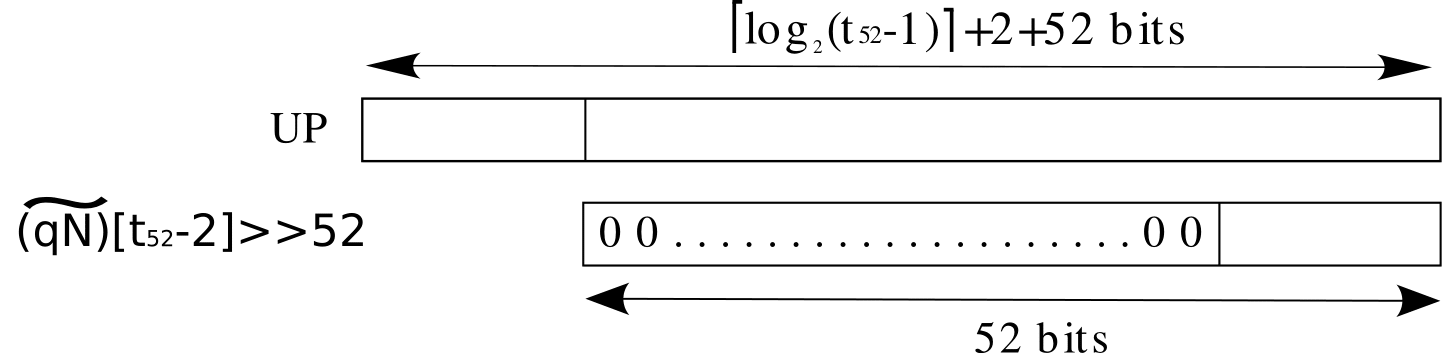}
    \caption{Detail of $UP + (\widetilde{(qN)}[t_{52}-2] >> 52)$}
    \label{fig:}
\end{figure}
Therefore, $\widetilde{(qN)}[t_{52}-2]>>52$ can be computed with the least significant bits of $UP$ and its sum with $UP$ does not generate any carry out of the $52th$ bit. Thus $\widetilde{(qN)}[t_{52}-1]>>52$ is computed only with the most significant bits of $UP$ and the carry $c_{add}$ which is propagated to the $53^{rd}$ bit of UP.
 \medskip
 
As a conclusion $ \widehat{(qN)}[0]$ is computed only with the partial products of weight at least equal $t_{52}-1$: 
\begin{equation}
\label{eq:qnhat}
	\begin{split}
    \widehat{(qN)}[0]  \gets & \! \sum^{t_{52}}_{i=0,j=0, i+j=t_{52}} \!\!\!\!\! \textrm{\texttt{mul52lo}}(q[i],N[j])  \\
     &+ \!\!\!\!\!\! \sum^{t_{52}-1}_{i=0,j=0, i+j=t_{52}-1} \!\!\!\!\! \textrm{\texttt{mul52hi}}(q[i],N[j])\\
     &+ \widetilde{(qN)}[t_{52}-1]>>52 + c_{add}.
	\end{split}
\end{equation}

\end{proof}

\subsection{Complexity of the Truncated Montgomery multiplication}
The truncated multiplication in Algorithm~\ref{alg:trunc_montred} at line 2 can be done in several ways.

\begin{table*}[htb]
  \caption{Batch Schoolbook truncated \texttt{B\_fma} used in Alg.~\ref{alg:8MontRed}, instruction number comparison}
  \label{tab:comptrfma}
  \begin{center}
  \small
  \begin{tabular}{|c|c|c|c|c|c|}
  	\hline
  			&	\# \texttt{VPMADD52}	&	\# shifts	&	\# Additions	& \# maskings	& \# OR\\
	\hline
	\texttt{B\_fma} &	$2 t_{52}^2$			&	$2 t_{52}-2$	&	$2 t_{52}-2$	&	$2 t_{52}-2$ &   -\\
	\hline
	\texttt{trunc. B\_fma}			&	$t_{52}^2 + 3/2 t_{52}-1$	&	$t_{52}$	&	$ t_{52}+2$	&	$t_{52}+2$    &  $t_{52}-1$ \\
	\hline
  \end{tabular}
  \end{center}
\end{table*}

\paragraph{Schoolbook multiplication.}
Here, we compute only the $t_{52}+1$ most significant words of the $q\times N$ product. The instruction count in Table~\ref{tab:comptrfma} shows that the instruction count is divided by nearly two for the truncated \texttt{B\_fma}, which replaces the operation line 2 Algorithme~\ref{alg:8MontRed}. 
Thus, for the whole Montgomery multiplication, this leads to a global complexity of slightly more than 2 multiplications, instead of 2.5 in the classical approaches, including the CIOS. 

\medskip

\paragraph{ Karatsuba multiplication.}

In this approach of the truncated multiplication, the operands have to be split with the following construction:

$$\begin{array}{cl}trunc(A\times B) =&2^t(D_{2\ell}+D_{1h} - D_{0h}-D_{2h})
    + 2^{3t/2}D_{2h}.
\end{array}$$

We therefore compute only $D_{2\ell}, D_{2h}, D_{1h}~$ and $~ D_{0h}$, that is one and two halves of elementary multiplications instead of three. The cost of this truncated multiplication is roughly two-thirds of a complete multiplication. And we need also half of a whole addition to achieve the last step of the Montgomery reduction.

The halves of elementary multiplications are computed using the same approach than the one for the Schoolbook. The correctness is ensured by the same kind of carry evaluation and propagation.

 Thus, in this case, for the whole Montgomery multiplication, this leads to a global complexity of a little more than 2.33 multiplications with one stage Karatsuba, instead of 2.0 for the truncated approaches with Schoolbook, or 2.67 for the conventional approaches. This values are respectively 2.56 and 2.78 for the two-stage Karatsuba.


\section{Performances of the implementations}
\label{sec:batchPerfs}

Our batch multiplications, batch Montgomery multiplications and the corresponding exponentiations have been implemented in C. All the source codes are available at \url{https://github.com/lea-gl/TruncatedBatchSIMDAVX512MontgomeryMultiplicationsModularExponentiation.git}.

In this section, after the presentation of the performance measurement procedure, we provide the results of the batch \texttt{AVX512} multiplications, Batch Montgomery multiplications and squarings, and corresponding exponentiations. This section concludes with the performances of the 256-bit batch exponentiation counterparts. All these experiments are compared to state-of-the-art \texttt{GMP} and \texttt{OpenSSL} performances, compiled and run on the same platform and using the same measurement procedure.

\subsection{Performances measurement procedure}

It is known that the intensive use of vectorized computation with \texttt{AVX2} and \texttt{AVX512} extensions can lead to penalties due to the resulting high power consumption \cite{throttling2017}. Intel processors have a limited power budget and may reduce their frequency when executing complex \texttt{AVX2} and \texttt{AVX512} instructions. This might also affect the execution of adjacent non-vectorized code. However, the impact of this issue depends on the power budget allowed by the processor and how efficiently it is cooled. Our goal here is to provide a fair evaluation of the sequential and our parallel algorithms. 
In order to perform this comparison as fairly as possible, we focused on measuring clock cycles.
Measurements were performed on a 
Dell Inspiron laptop with a \texttt{tiger lake} processor.
\begin{verbatim}
vendor_id	: GenuineIntel
cpu family	: 6
model		: 140
model name	: 11th Gen Intel(R) Core(TM) i7-1165G7 @ 2.80GHz
\end{verbatim}
\normalsize
The compiler is \texttt{gcc} version 9.4.0, the compiler options are as follows:

\noindent\texttt{-O3 -funroll-all-loops -g -march=native -lgmp  -lcrypto}.

\noindent We kept the \texttt{-funroll-all-loops} option though it does not provide significant improvements. We follow the same kind of test procedure as described in~\cite{NDruckerGK18} or in~\cite{RobertV22}:
\begin {itemize}
	\item the \textit{Turbo-Boost}\textregistered~is deactivated during the tests;
	\item 1000 runs are executed in order to "warm-up" the cache memory;
	\item 50 random data sets are generated, and for each data set
      the minimum of the execution clock cycle numbers over a batch of 1000 runs is recorded;
	\item the performance is the average of all these minimums;
\end{itemize}

The clock cycles have been counted with \texttt{rdtsc/rdtscp} instructions.

\subsection{Performances of the batch multipliers}

We have implemented both batch squaring and multiplication for several sizes. The timings are shown in Table~\ref{tab:perfmulV2} where the best results are in bold. We target operands having between 260 and 4154 bits. Due to the Karatsuba splitting, the size of Karatsuba implementations are slightly different. Karatsuba multiplication is inefficient for small operands so we have not implemented it for 260-bit operands. The 4154-bit operands are large enough to permit a double Karatsuba splitting.

We have compared our implementations with 8 successive \texttt{GMP} low level multiplications \texttt{mpn\_mul\_n()}. It can be observed that the batch approach takes advantage of vector instructions for both squaring and multiplication. The batch approach is 5 to 9 times faster than \texttt{GMP}. For large enough operands, the Karatsuba approach is the fastest.

\begin{table*}[ht]
	\caption{Batch Schoolbook (SB) and Karatsuba multiplications, clock cycles.}
	\label{tab:perfmulV2}

		\begin{center}
            \small
			\begin{tabular}{|*{6}{c|}}
				\hline
				op. size (bits)	SB			&	260	&	520	&	1040	&	2080	&	4108	\\
				op. size (bits)	Karatsuba				&		&	518	&	1038	&	2078	&	4154	\\
				\hline
				\hline
				\multicolumn{6}{|c|}{\# clock cycles}\\
				\hline
				\texttt{GMP} \texttt{mpn\_mul\_n()} ($\times 8$)	&	574	&	1526	&	5213	&	16361	&	50099	\\
				\hline
				\texttt{GMP} \texttt{mpn\_sqr()} ($\times 8$)	&	376	&	967	&	2994	&	9871	&	31322	\\
				\hline
				\hline
				Mul. Schoolbook					&	\bf 61	&	\bf 223	&	890	&	4076	&	16326	\\
				\hline
				Mul. Karatsuba				&		&	258	&	\bf 849	&	\bf 3054	&	\bf 10367	\\
				\hline
				\hline
				Squaring Schoolbook				&	\bf 49	&	\bf 149	&	\bf 501	&	\bf 1898	&	8923	\\
				\hline
				Squaring Karatsuba			&		&	220	&	622	&	1964	&	\bf 6882	\\
				\hline
			\end{tabular}
			
		\end{center}
\end{table*}

A few comments on the results:

\begin{itemize}
    \item Concerning the multiplication, the Schoolbook Algorithm~\ref{alg:8SBmul} is better than its Karatsuba counterpart for the 260 and 520-bit sizes. This is consistent with the complexities
    \item Concerning the squaring, one can see that the Karatsuba approach is better only in the 4154-bit case. This is explained by the fact that the Schoolbook squaring (Algorithm~\ref{alg:8SBsqu}) makes intensive use of the \texttt{madd52lo} and \texttt{madd52hi} instructions with no use of additions except for the carry management, whereas in the Karatsuba case, except for the elementary squarings, the final reconstruction requires a lot of additions which can not be saved in the same way as in Schoolbook case. Thus, the threshold for better efficiency in the Karatsuba case is much higher in comparison with the multiplication case.
\end{itemize}

\subsection{Performances of the batch Montgomery multiplication and squaring}

The performances of the batch Montgomery multiplications and squarings are given in Table~\ref{tab:perfmontmul}, for the classical and truncated versions. We evaluated our implementations for $1024$, $2048$ and $4096$-bit operands.

Because of the interleaved word-size multiplications, Algorithm~\ref{alg:8BlockMontRed} has only been implemented with the Schoolbook multiplication. The $1024$-bit Karatsuba multiplication is slower than the Schoolbook multiplication. As a consequence, we have not implemented Truncated Karatsuba versions for this size.  

For comparison sake, we implemented Montgomery modular multiplication and squaring using \texttt{GMP} \texttt{mpn} Montgomery operations, and \texttt{OpenSSL BN\_mod\_mul\_montgomery} functions. However, we have not found a specific \texttt{OpenSSL} counterpart of the Montgomery modular squaring. The best timings are in bold in Table~\ref{tab:perfmontmul}.
The best implementations are more than 4 times faster than \texttt{OpenSSL} in all cases. Furthermore, the truncated approach is about 21~\% faster than the conventional Schoolbook Montgomery multiplication, and up to nearly 28~\% faster than the 4108-bit squaring.

For the Karatsuba approaches, the truncated version is almost 21~\% faster than the classical multiplication for $2048$-bit operands and up. In the case of the squaring the improvement is about 15~\%. Even for the largest size (4154 bits), the Truncated Karatsuba modular squaring remains slightly slower than its Schoolbook counterpart, while the Karatsuba modular multiplication is nearly 10 \% faster than the Schoolbook version.

\begin{table*}[htp]
\caption{Batch Montgomery multiplications and squarings (normal and truncated) \# clock cycles.}
\label{tab:perfmontmul}
\small
\begin{tabular}{|c|c||*{3}{c|}|*{3}{c|}|}
				\hline
			\multicolumn{2}{|c||}{ } &	\multicolumn{3}{|c||}{Montgomery  }	&	\multicolumn{3}{|c||}{Montgomery  }\\
			\multicolumn{2}{|c||}{ } &	\multicolumn{3}{|c||}{ modular }	&	\multicolumn{3}{|c||}{ modular }\\
			\multicolumn{2}{|c||}{ } &	\multicolumn{3}{|c||}{  multiplications}	&	\multicolumn{3}{|c||}{  squarings}\\
				\hline
				\multicolumn{2}{|c||}{Modulus Size}	&	 $1024$	&	 $2048$	&	 $4096$	&	 $1024$	&	 $2048$	&	 $4096$\\
				\hline
				\hline
				\multicolumn{2}{|c||}{{ \texttt{GMP} } ($\times 8$)}
				
					&	9862	&	35497	&	120342	&	8358	&	30465	&	101373	\\
				\hline
				\multicolumn{2}{|c||}{{ \texttt{OpenSSL}} ($\times 8$)}     
     &	7949	&	29928	&	116898	&	-	&	-	&	-		\\
				\hline
				\hline
				\bf Batch & \bf Multiplication & \multicolumn{6}{|c||}{\bf This work}\\
    \hline
    \hline
               		Algorithm~\ref{alg:8MontRed} & Schoolbook  &	2276	&	8696	&	38609	&	1885	&	7154	&	32107\\
				\hline
				Algorithm~\ref{alg:8BlockMontRed}	& Schoolbook  &	2162	&	8936	&	42730	&	-	&	-	&	-	\\
                \hline
                	 Algorithm~\ref{alg:8MontRed}	&	 Truncated Schoolbook & \bf	{1847}	&	7306	&	30492	&	\bf	{1439}	&	\bf	{5548}	&	\bf	{23413}   \\
                  \hline
				\multicolumn{2}{|c||}{speedup vs \texttt{GMP}}	&	5.34	&	4.86	&	3.95	&	5.81	&	5.49	&	4.33	\\	
                  \hline
				\multicolumn{2}{|c||}{speedup vs \texttt{OpenSSL}}	&	4.31	&	4.10	&	3.69	&	-	&	-	&	-	\\	
               \hline
               \hline
                Algorithm~\ref{alg:8MontRed} & Karatsuba		&	2423	&    8400	&  34628	&	2191	&	7330	&	28841	\\
				\hline
                	 Algorithm~\ref{alg:8MontRed}	& Truncated Karatsuba &	-	&\bf{7286}	&\bf {27594}	&	-	&	6202	&	23736		\\
                  \hline
				\multicolumn{2}{|c||}{speedup vs \texttt{GMP}}	&	4.07	&	4.87	&	4.36	&	3.81	&	4.91 &	4.27		\\	
                  \hline
				\multicolumn{2}{|c||}{speedup vs \texttt{OpenSSL} }	&	3.28 &	4.12	&	4.24	&	-	&	-	&	-	\\	
				\hline
			\end{tabular}
\end{table*}

\subsection{Window exponentiation with Truncated Montgomery Modular Mutiplication}

We have implemented Left-to-Right fixed-window exponentiation in a constant time fashion for $1024$, $2048$ and $4096$-bit operands, i.e. the modulus size. These implementations make use of the batch Montgomery squarings, \texttt{B\_fma}s and multiplications mentioned above.

\begin{algorithm}[ht]
  \caption{Constant-time Batch Fixed-Window Left-to-Right Exponentiation}
  \label{alg:8FWExp}
  \small
  \begin{algorithmic}[1]
    \Require  Eight values $a_k$, the corresponding eight $s$-bit exponents and moduli $e_k$ and $m_k$, all stored in 64-bit word arrays, the window width $w$.
    \Ensure The eight modular exponentiations $y_k = a_k^{e_k} \bmod m_k$
    
    //batch of 8 values $A_k$ stored in 52-bit slices in $t_{52}$ \mavxi shares
    \State $A_k\gets$\texttt{expand}$(a_0,\hdots,a_7)$

    //batch of 8 moduli $M_k$ stored in 52-bit slices in $t_{52}$ \mavxi shares
    \State $M_k\gets$\texttt{expand}$(m_0,\hdots,m_7)$

    //batch of 8 exponents $E_k$ stored in 64-bit slices in $t_{64}$ \mavxi shares
    \State $E_k\gets$\texttt{expand}$_{64}(e_0,\hdots,e_7)$

    \For{$k$ from 0 to 7} {\bf in parallel}
        \State $Y_k\gets \mathbf 1$ // batch of ones
        \For{$i$ from 0 to $2^w$} //precomputation
            \State  $G_k[i] \gets A_k^i \bmod M_k$
        \EndFor
        \For{$i$ from $s-2w$ to 0 by $w$}// main loop
            \State $b_k\gets E_k[i,i+w-1]$ // $w$ bits of $E_k$
            \State $tmp_k \gets G_k(b_k)$ // constant-time batch selection
            \State $Y_k \gets Y_k^{2^w} \bmod M_k$ // $w$-batch Montgomery squarings
            \State $Y_k\gets Y_k\times tmp_k \bmod M_k$
        \EndFor
        \State loop epilog if necessary
    \EndFor

    //backward conversion of the 8 results in 64-bit word arrays
    \State $y_k\gets$ \texttt{contract}$(Y_k)$
    \State return the eight results $y_k = a_k^{e_k}\bmod m_k$
    
    \end{algorithmic}
\end{algorithm}

They aim to compare our approach for the modular operations with state-of-the-art modular exponentiations. The exponentiation functions take as arguments big integers represented by 64-bit word arrays, identical to those used in the low-level functions of the \texttt{GMP} library. The result is stored in the same fashion.

Thus, the exponentiation is processed as follows:

\begin{itemize}
    \item conversion of a batch of operands stored in 64-bit word arrays in word-slicing representation.
    \item computation of the batch exponentiation
    \item backward conversion to a batch of results stored in 64-bit word arrays.
\end{itemize}

This conversion is implemented in conventional C using maskings and shiftings. Since the complexity is linear in the operand size, the cost remains negligible. Nevertheless, we provide in Table~\ref{tab:batchconversion} the timings of our implementations in clock cycles. The forward conversion from the conventional representation to the batch 52-bit representation is called \texttt{expand} and the backward conversion is called \texttt{contract}.

\begin{table*}[htp]
        \caption{Batch conversions, $\# 10^3$ clock cycles.}
	    \label{tab:batchconversion}
    \begin{center}
        \small
        \begin{tabular}{c|c|c|c|c|}
        
        \cline{2-5}
        &   size    &  1040    &   2080       &    4108    \\
        \hline
        \multirow{2}{*}{\# clock cycles}   &   
          \texttt{expand}   &  1165   &   2411  &   4790\\
         \cline{2-5}
          &\texttt{contract}   & 1726 &   3274   &   6505\\
        \hline
        \end{tabular}
    \end{center}
\end{table*}
We tested window sizes from 1 (L-R Square-and-Multiply-always) to 5. The best window size is 4 for $1024$-bit moduli and 5 in the other cases. 
We used the fastest modular squaring and multiplication. For $1024$-bit moduli, we used Algorithm~\ref{alg:8BlockMontRed}. In the other cases, we used Algorithm~\ref{alg:8MontRed}. This approach is shown Algorithm~\ref{alg:8FWExp}.

\subsubsection{Experimentation of \texttt{AVX512} versions}

The timings are summarized in the last three columns of Table~\ref{tab:perfmontexp}, for the implemented modulus sizes (1024, 2048 and 4096 bits). 

Whatever the multiplication used (Schoolbook or Karatsuba), the truncated version is always faster than the classical version. The improvement ranges from 13\% faster with the $4096$-bit Truncated Karatsuba, to  20~\% faster with the $1024$-bit Truncated Schoolbook version.
As expected, the Karatsuba approach is only better for the largest size of $4096$ bits. 

We compare these results with \texttt{GMP}  and \texttt{OpenSSL} libraries, providing the timings for eight modular exponentiation computations.
\begin{itemize}
    \item The \texttt{GMP} version is the 6.2.0 and we used the function \texttt{mpn\_sec\_powm}, which is specifically designed for cryptographic use.
    \item The \texttt{OpenSSL} version is the 3.2.1. This version is compiled on our platform and provides the \texttt{RSA\_Z} operations, which make use of the \texttt{AVX2} instruction set. We measured the performance of two functions:
    \begin{itemize}
        \item \texttt{BN\_mod\_exp\_mont\_consttime}, which is the \texttt{RSA\_Z} implementation (\texttt{AVX2} version, see~\cite{gueron2012software}).
        \item \sloppypar{\texttt{BN\_mod\_exp\_mont\_consttimex2}, which computes 2 exponentiations simultaneously. For 1024 and 2048-bit operands, this function implements the \texttt{VPMADD52} instructions, but in 256-bit registers. For the 4096-bit operands, this function falls back on two calls of \texttt{BN\_mod\_exp\_mont\_consttime}.} This explains why we do not present 4096-bit performance results for this function Table~\ref{tab:perfmontexp}.
    \end{itemize}
    These functions are constant-time fixed-window exponentiations, and implement a CIOS-like Montgomery modular multiplication and squaring.
\end{itemize}

The \texttt{OpenSSL} functions offers better results than \texttt{GMP} ones.

\sloppypar{Our implementation uses a constant time fixed-window approach similar as the one of the \texttt{OpenSSL} functions. Compared to \texttt{OpenSSL} \texttt{BN\_mod\_exp\_mont\_consttime}, the best speedup of our implementations is achieved for the $1024$-bit operands, with almost 4 times fewer cycles per exponentiation (our Truncated Schoolbook). The best other cases provide speedups of 3.71 and 3.37, respectively, for the $2048$-bit Schoolbook and the  $4096$-bit Karatsuba.}

Compared to \texttt{OpenSSL} \texttt{BN\_mod\_exp\_mont\_consttimex2}, the best speedup of our implementations is achieved for the $1024$-bit operands, with 1.75 times fewer cycles per exponentiation (our Truncated Schoolbook). The best other cases provide speedups of 1.38 and 1.27, respectively, for the $2048$-bit Schoolbook and Karatsuba.

One may notice that the non-truncated versions remain better than their \texttt{RSA\_Z} \texttt{BN\_mod\_exp\_mont\_consttime} or \texttt{BN\_mod\_exp\_mont\_consttimex2} counterparts.

\begin{table}[htp]
	\caption{Batch of 8 modular fixed-window exponentiations, $\# 10^3$ clock cycles.}
	\label{tab:perfmontexp}
		\begin{center}
		\small
			\begin{tabular}{|c|c||*{3}{c|}|}
				\hline
			\multicolumn{5}{|c||}{\large \bf $\times 8$ Modular Fixed-Window exponentiation ( $\times10^3$\#cc)}\\
				\hline
				\multicolumn{2}{|c||}{Modulus Size}	&	 $1024$	&	 $2048$	&	 $4096$	\\
				\hline
				\hline
				\multicolumn{2}{|c||}{{ \texttt{GMP} } ($\times 8$)}
				
					&	11333	&	81288	&	614637\\
				\hline
				\multicolumn{2}{|c||}{{ \texttt{OpenSSL}  \texttt{BN\_mod\_exp\_mont\_consttime}} ($\times 8$)}    
    &	8313	&	58445	&	434359	\\
				\hline
				\multicolumn{2}{|c||}{{ \texttt{OpenSSL}  \texttt{BN\_mod\_exp\_mont\_consttimex2}} ($\times 4$)}   
    &	3648	&	21674	&	442533	\\
				\hline
				\hline
				\bf Batch & \bf multiplication & \multicolumn{3}{|c||}{\bf This work}\\
    \hline
    \hline
               		Algorithm~\ref{alg:8MontRed} \&~\ref{alg:8BlockMontRed} & Schoolbook  &    2589    &   19678   &   177095  \\

                \hline
                	 Algorithm~\ref{alg:8MontRed}	&	 Truncated Schoolbook &	\bf\underline{2090}	&\bf\underline{15736}	&	131753   \\
                  \hline
				\multicolumn{2}{|c||}{speedup vs \texttt{GMP}}	&	5.42	&	5.17	&	4.66\\	
                  \hline
				\multicolumn{2}{|c||}{speedup vs \texttt{OpenSSL} \texttt{BN\_mod\_exp\_mont\_consttime}}	&	3.98	&	3.71	&	3.30\\	
                  \hline
				\multicolumn{2}{|c||}{speedup vs \texttt{OpenSSL} \texttt{BN\_mod\_exp\_mont\_consttimex2}}	&	1.75 &	1.38	&	-\\	
               \hline
               \hline
                Algorithm~\ref{alg:8MontRed} & Karatsuba		&	-	&	19717	&	144567\\
				\hline
                	 Algorithm~\ref{alg:8MontRed}	& Truncated Karatsuba &	-	&	17124	&	\bf\underline{129015}	\\
                  \hline
				\multicolumn{2}{|c||}{speedup vs \texttt{GMP}}	&	-	&	4.75	&	4.76	\\	
                  \hline
				\multicolumn{2}{|c||}{speedup vs \texttt{OpenSSL} \texttt{BN\_mod\_exp\_mont\_consttime}}	&	-  &3.41	&	3.37	\\	
                  \hline
				\multicolumn{2}{|c||}{speedup vs texttt{OpenSSL} \texttt{BN\_mod\_exp\_mont\_consttimex2}}	&	-  &  1.27	&	-	\\	
				\hline
			\end{tabular}
		\end{center}
\end{table}

\subsubsection{Experimentation of 256-bit versions}

In order to provide a fair comparison with the \texttt{BN\_mod\_exp\_mont\_consttimex2} \texttt{OpenSSL} implementation using the 256-bit vectorized fused multiplier-adder (\texttt{\_mm256\_madd52*\_epu64(a, b, c)}), we derived 256-bit versions from the previous implementations, computing a batch of four values instead of eight for 512-bit versions.

\paragraph{Comparison between 256-bit and 512-bit implementations.}

We discuss here the register size impact. In our implementations, the difference between both 512-bit and 256-bit versions is in the registers used: respectively \texttt{zmm} and \texttt{ymm}. Thus, one might expect the retired instruction number to be the same between the two versions. However, this is not the case for the clock cycle numbers, since the instruction throughput in the 256-bit case is much lower than that of the 512-bit instructions. In order to check the performance level of both versions, we provide the clock cycle numbers per exponentiation of both versions. In other words, this is the batch delay divided by the number of operations in the batch, 8 and 4 respectively (Table~\ref{tab:ifma256-512comp}).

\begin{table}[htp]
	\caption{Batch Modular fixed-window exponentiations, $\# 10^3$ clock cycles per exponentiation.}
	\label{tab:ifma256-512comp}
		\begin{center}
		\small
			\begin{tabular}{|c|c||*{2}{c|}|}
				\hline
			\multicolumn{4}{|c||}{\large \bf Modular Fixed-Window exponentiation}\\
				\hline
				\multicolumn{2}{|c||}{\multirow{2}{*}{Register Type}}	&	512-bit 	&	 	256-bit	\\
                \multicolumn{2}{|c||}{}	&$\times10^3$\#cc/8&$\times10^3$\#cc/4\\
				\hline
				\hline
				\bf Mod.size & \bf Multiplication & \multicolumn{2}{|c||}{\bf This work}\\
    \hline
    \hline
               		\multirow{2}{*}{1024 bits} & Schoolbook  &    325    &   419     \\
                \cline{2-4}
                	 &	 Truncated Schoolbook &	{263}	&{301}  \\
    \hline
    \hline
               		\multirow{2}{*}{2048 bits} & Schoolbook  &   2401     &    2918    \\
                \cline{2-4}
                	 &	 Truncated Schoolbook &	{1965}	&{2583}  \\
               \hline
			\end{tabular}
		\end{center}
\end{table}

In any case, the clock cycle number per exponentiation is better in the \texttt{AVX512} configurations. The advantage ranges from 14.4\% (Truncated - 1024 bits) to 31.5 \% (Truncated - 2048 bits) while it is between 20 and 30 \% in the other configurations.

\paragraph{Comparison with \texttt{OpenSSL BN\_mod\_exp\_mont\_consttimex2} function.}~\\
Since the \texttt{BN\_mod\_exp\_mont\_consttimex2} function provides vectorized computations for sizes from 1024 to 2048 bits, we provide  the comparison between our work and the \texttt{OpenSSL} function for these sizes (Table~\ref{tab:ifma256FWExp}). We compare our batch implementation, which computes 4 exponentiations simultaneously, with two successive runs of the \texttt{OpenSSL BN\_mod\_exp\_mont\_consttimex2}.

\begin{table}[htp]
	\caption{Batch of 4 Modular fixed-window exponentiations, $\# 10^3$ clock cycles.}
	\label{tab:ifma256FWExp}
		\begin{center}
		\small
			\begin{tabular}{|c|c||*{3}{c|}|}
				\hline
			\multicolumn{5}{|c||}{\large \bf $\times 4$ Modular Fixed-Window exponentiation ( $\times10^3$\#cc)}\\
				\hline
				\multicolumn{2}{|c||}{\multirow{2}{*}{Modulus Size bits}	} &	 \multirow{2}{*}{$1024$}	&	 \multirow{2}{*}{$2048$}		& ratio\\
                \multicolumn{2}{|c||}{}	&&&2048/1024\\
				\hline
				\hline
				\multicolumn{2}{|c||}{{ \texttt{OpenSSL}  \texttt{BN\_mod\_exp\_mont\_consttimex2}} ($\times 2$)}    
    &	1824  &	10837	&   5.94\\
				\hline
				\hline
				\bf Batch & \bf Multiplication & \multicolumn{3}{|c||}{\bf This work}\\
    \hline
    \hline
               		Algorithm~\ref{alg:8MontRed} \&~\ref{alg:8BlockMontRed} & Schoolbook  &    1678    &   11632     & 6.96\\
                \hline
                	 Algorithm~\ref{alg:8MontRed}	&	 Truncated Schoolbook &	\bf\underline{1203}	&\bf\underline{10290}  &   8.59\\
                  \hline
				\multicolumn{2}{|c||}{speedup truncated vs \texttt{BN\_mod\_exp\_mont\_consttimex2}}	&	\bf 1.52 &	\bf 1.05	&\\	
               \hline
			\end{tabular}
		\end{center}
\end{table}

Some comment on this table about the 256-bit versions of our modular exponentiation:

\begin{itemize}
    \item The best speedup is achieved for the 1024-bit modulus size. Our implementation using the Truncated Montgomery reduction provides a speedup of 1.51 (window width = 5).
    \item Our 1024-bit modulus implementation using the conventional Montgomery reduction is also slightly faster than the \texttt{BN\_mod\_exp\_mont\_consttimex2} function, with a 1.09 speedup (window width = 4).
    \item For the 2048-bit modulus, our implementation using the conventional Montgomery reduction is slightly slower than the \texttt{BN\_mod\_exp\_mont\_consttimex2} function. We remind here that for a square-and-multiply variant algorithm based implementation, the theoretical ratio should be 8 between the 2048-bit and the 1024-bit versions. These ratios are indicated in the right column of Table~\ref{tab:ifma256FWExp}. One can see that this ratio is better in the \texttt{OpenSSL} case, especially compared to our versions using the Truncated Montgomery reduction. 
    As a conclusion, compared to the \texttt{BN\_mod\_exp\_mont\_consttimex2} function, while our implementation using the conventional Montgomery reduction is slightly slower by around 7 \% (window width = 5), the one using our truncated approach for the modular reduction gives a 1.05 speedup (window width = 4).
\end{itemize}

As a conclusion, this shows the potential of the Truncated Montgomery reduction applied to various implementation situations. In every case explored here, our proposed approach is, to the best of our knowledge, faster than the state-of-the-art implementations.

\section{Conclusion}

In this paper, we present new software implementations using \texttt{AVX512} instruction set and taking advantage of the \texttt{VPMADD52} instructions, which compute a vectorized fused multiplication-addition. We implemented multi-precision multiplications and squarings, for sizes from 260 to 4154 bits. We used these implementations in Montgomery modular multiplications and squaring along with CIOS Montgomery multiplications. We also present a new approach of Truncated Montgomery multiplication computing the most significant higher half part of one of the multiplications involved in the Montgomery modular reduction in order to speedup the computation. Our implementations of this approach are more than 4 times faster than the \texttt{OpenSSL} ones. Moreover, used in fixed-window exponentiations of sizes $1024$, $2048$ and $4096$ bits, compared to \texttt{BN\_mod\_exp\_mont\_consttime}, the best speedups are respectively 3.98, 3.71, 3.37 for our implementations using our new Truncated Schoolbook or Karatsuba Montgomery modular multiplications. Compared to \texttt{BN\_mod\_exp\_mont\_consttimex2} using \texttt{madd52*} in 256-bit registers, in fixed-window exponentiations of sizes $1024$ and $2048$, our \texttt{AVX512} implementations provide speedups of 1.75 and 1.38 respectively, while their 256-bit counterparts give speedups of 1.51 and 1.05 for $1024$ and $2048$-bit sizes (batch of 4 values in this case).

The speedups are good because the batch operations are highly parallel and can be easily vectorized. Similar results could be obtained on other processors whose architecture has SIMD instruction sets. The NEON instruction set on ARM processors \cite{neonintrinsics} offers such possibilities. However, these instructions are not exactly identical to \texttt{AVX512} extensions. For example, \texttt{VPMADD52} instructions have only a 32-bit equivalent on this architecture. This could lead to a different word slicing of the operands. Very significant gains over sequential implementations can be reach, but may not be the same to those with \texttt{AVX512} instructions.

\paragraph{Perspectives.}
The improvements presented in this work could be adapted to other contexts as well.  Batch computations have demonstrated their value and could be utilized, for instance, in post-quantum schemes.  While our study of Truncated Montgomery multiplication was focused on RSA, there are other potential applications worth studying. Indeed, schemes relying on supersingular elliptic curves like pairing-based cryptography (see \cite{mrabetJ2017} and \cite{DuquesneL2006}) or isogeny based post-quantum protocols (see \cite{FeoKLPW2020}), also require large integer modular multiplications.

In addition, homomorphic encryption protocols based on large integers may take advantage of our approach, among them, Coron \emph{et al.} \cite{CoronMNT2010} and Dyer \emph{et al.} \cite{DyerDX2019}.


\bibliography{biblio.bib}

\end{document}